\documentclass[copyright,creativecommons]{eptcs}
\providecommand{\event}{SYNT 2013} 

\usepackage{amsmath,amsthm,amsfonts,alltt}
\usepackage{local}

\usepackage{algorithm}
\usepackage[noend]{algpseudocode}
\usepackage{paralist}

\usepackage{listings}
\usepackage{tikz}

\newtheorem{theorem}{Theorem}

\newtheorem{lemma}[theorem]{Lemma}
\newtheorem{proposition}[theorem]{Proposition}

\renewcommand{\next}{\mathit{next}}
\newcommand{\comment}[1]{}

\def\titlerunning{Toward Synthesis of Network Updates}
\def\authorrunning{Noyes, Warszawski, {\v C}ern\'y, and Foster}

\title{\titlerunning}

\author{
Andrew Noyes
\institute{Cornell University}
\and
Todd Warszawski 
\institute{Cornell University} 
\and
Pavol {\v C}ern\'y
\institute{University of Colorado Boulder} 
\and
Nate Foster
\institute{Cornell University} 
}

\begin{document}

\maketitle

\begin{abstract}
Updates to network configurations are notoriously difficult to
implement correctly. Even if the old and new configurations are
correct, the update process can introduce transient errors such as
forwarding loops, dropped packets, and access control violations. The
key factor that makes updates difficult to implement is that networks
are distributed systems with hundreds or even thousands of nodes, but
updates must be rolled out one node at a time. In networks today, the
task of determining a correct sequence of updates is usually done
manually---a tedious and error-prone process for network
operators. This paper presents a new tool for synthesizing network
updates automatically. The tool generates efficient updates that are
guaranteed to respect invariants specified by the operator. It works
by navigating through the (restricted) space of possible solutions,
learning from counterexamples to improve scalability and optimize
performance. We have implemented our tool in OCaml, and conducted
experiments showing that it scales to networks with a thousand
switches and tens of switches updating.
\end{abstract}

\section{Introduction}

Most networks are updated frequently, for reasons ranging from taking
devices down for maintenance, to modifying forwarding paths to avoid
congestion, to changing security policies.  Unfortunately,
implementing a network update correctly is an extremely difficult
task---it requires modifying the configurations of hundreds or even
thousands of routers and switches, all while traffic continues to flow
through the network. Implementing updates naively can easily lead to
situations where traffic is processed by switches in different
configurations, leading to problems such as increased congestion,
temporary outages, forwarding loops, black holes, and security
vulnerabilities.

The research community has developed a number of mechanisms for
implementing network updates while preserving important
invariants~\cite{francois:07,francois-bgp,consensus,raza:11,vanbever:11,frenetic-consistent-updates}.
For example, \emph{consensus routing} uses distributed snapshots to
avoid anomalies in routing protocols such as
BGP~\cite{consensus}. Similarly, \emph{consistent updates} uses
versioning to ensure that every packet traversing the network will be
processed with either the old configuration or the new configuration,
but not a mixture of the two~\cite{frenetic-consistent-updates}. But
these mechanisms are either limited to specific protocols and
properties, or are general but expensive to implement, requiring
substantial additional space on switches to represent the forwarding
rules for different configurations.

This paper explores a different idea: rather than attempting to design
a new concrete update mechanism, we use synthesis to generate such
mechanisms automatically. With our system, the network operator
provides the current and target configurations as input, as well as a
collection of invariants that are expected to hold during the
transition. (The current and target configurations should also satisfy
these invariants.) The system either (i) generates a sequence of
modifications to the forwarding rules on individual switches that
transitions the network to the new configuration and preserves the
specified invariants, or (ii) halts with a failure if no such sequence
exists. Overall, our system takes a challenging programming task
usually done by hand today and automates it, using a back-end solver
to perform all tedious and error-prone reasoning involving low-level
network artifacts.

Our system provides network operators with a general and flexible tool
for specifying and implementing network updates efficiently. By
enabling them to specify just the properties that are needed to ensure
correctness, the synthesized updates are able to make use of
mechanisms that would be ruled out in other systems. For example, if
the operator specifies no invariants, then the tool can simply update
the switches in any order, without worrying about possible ill-effects
on in-flight packets. Alternatively, if the operator specifies an
invariant that encodes a firewall, then the network may forward
packets along paths that are different than the ones specified by the
old and new policies, as long as all packets blocked by the firewall
are dropped. This flexibility gives our system substantial latitude in
generating update implementations, and allows it to generate efficient
updates that converge faster, or use fewer forwarding rules, compared
to general techniques such as consistent updates.

Operationally, our system works by checking network properties using a
model checker. We encode the configuration of each switch into the
model, as well as the contents and location of a single in-flight
packet. Using this model, we then pose a sequence of queries to the
model checker, attempting to identify a modification to some switch
configuration that will transition the network to a more updated state
without violating the specified invariants. Determining whether a
configuration violates the invariants is a straightforward LTL model
checking problem. If this step succeeds, then we recurse and continue
the process until we eventually arrive at the new
configuration. Otherwise, we use the counterexample returned by the
model checker to refine our model and repeat the step.

We are able to reduce our synthesis problem to a reachability problem
(as opposed to a game problem) because we assume that the environment
is stable during the time the updates are performed. That is, we
assume that switches do not come up or go down, and that no other
updates are being performed simultaneously. The key challenge in our
setting stems from the fact that although the individual switch
modifications only need to maintain a correct overall configuration,
network configurations are rich structures, so navigating the space of
possible updates effectively is critical. We plan to investigate the
game version of our synthesis problem in future work.

In summary, this paper makes the following contributions
\begin{itemize}
\item We present a novel approach to specifying and implementing
  network updates using synthesis.
\item We develop encodings and algorithms for automatically
  synthesizing network updates using a model checker, and
  optimizations that improve its scalability and performance.
\item We describe a prototype implementation and present the results
  of experiments demonstrating that even our current prototype tool is
  able to scale to networks of realistic size.
\end{itemize} 
The rest of this paper is structured as follows.
Section~\ref{sec:overview} provides an overview to the network update
problem and discusses examples that illustrate the challenges of
synthesizing updates. Section~\ref{sec:model} develops our abstract
network model and defines the update problem
formally. Section~\ref{sec:algorithms} presents algorithms for
synthesizing network updates. Section~\ref{sec:implementation}
describes our implementation. Section~\ref{sec:experiments} presents
the results of our experiments. Section~\ref{sec:related} discusses
related work. We conclude in Section~\ref{sec:conclusion}.

\section{Overview}
\label{sec:overview}

\begin{figure}[t]
\centerline{\tikz\node[scale=.65]{\pgfimage{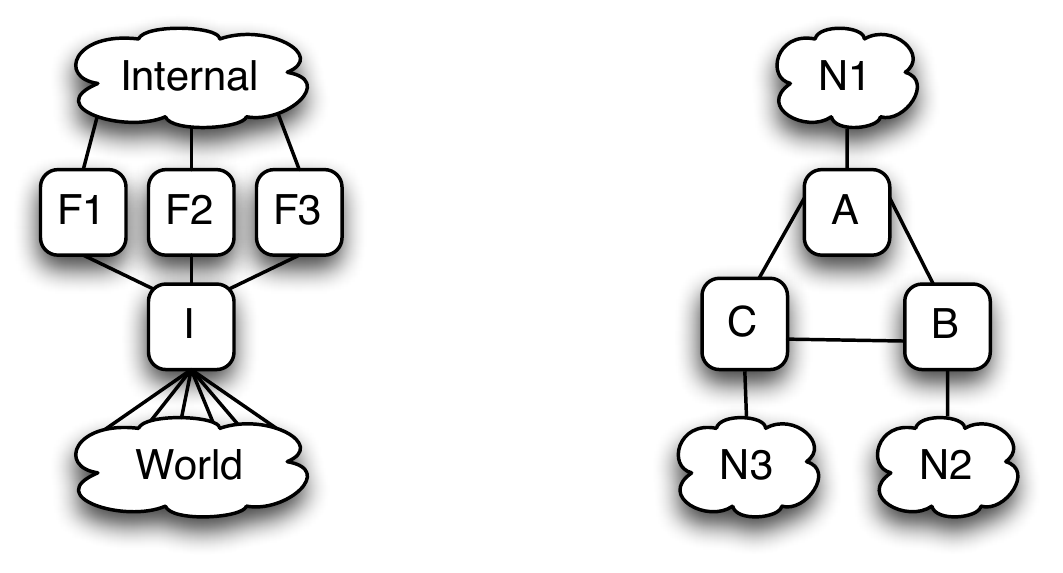}};}
\caption{Example network topologies: (a) distributed firewall, (b) cycle.}
\label{fig:example}
\end{figure}

This section provides a basic overview of primitive network update
mechanisms, and presents examples that illustrate the inherent
challenges in implementing network updates.

\paragraph*{Basics.}
Abstractly, a network can be thought of as a graph with switches as
nodes and links as edges. The behavior of each switch is determined by
a set of forwarding rules installed locally. A forwarding rule
consists of a pattern, which describes a set of packets, and a list of
actions, which specify how packets matching the pattern should be
processed. For the purposes of this paper, the precise capabilities of
patterns and actions and the details of how they are represented on
switches will not be important. However, typically patterns support
matching on packet headers and actions support optionally modifying
those headers and forwarding packets out one of its ports.

To process a packet, the network interleaves steps of processing using
the rules installed on switches and steps of processing using the
links themselves. More specifically, given a packet located at a
particular switch, the switch finds a matching rule and applies its
actions to the packet. This moves the packet to an output port on the
switch (or drops it). Assuming there is a link connected to that port,
the network will then transmit the packet to the adjacent switch, and
processing continues.

A network property is a set of paths through the topology. Such
properties can be used to capture basic reachability properties such
as connectivity and loop freedom, as well as more intricate properties
such as access control.

To implement an update to a new configuration, the operator issues
commands that install or uninstall individual forwarding rules on
switches. By carefully constructing sequences of commands, it is
possible to implement an atomic update on a single switch. For
example, the operator can install a set of new rules at a lower
priority than the current rules, and then delete the current rules
using a single uninstall command. But it is not possible to implement
simultaneous coordinated updates to multiple switches, as the network
is a distributed system.

\paragraph*{Distributed Firewall.}
We now present some simple examples of networks and updates that would
be difficult to implement by hand, as motivation for the synthesis
tool described in the following sections. The first example is a
variant of one originally proposed by Reitblatt et
al.~\cite{frenetic-consistent-updates}. The network topology, shown in
Figure~\ref{fig:example}~(a), consists of an ingress switch $I$ and
three filtering switches $F_1$, $F_2$, $F_3$. For simplicity, assume
that traffic flows up from the ``world'' to the ``internal''
network. At all times, the network is required to implement the
following security policy: (i) traffic from authenticated hosts is
allowed, (ii) web traffic from guest hosts is allowed, but (iii)
non-web traffic from guest hosts is blocked.

Initially the network is configured so that the ingress switch $I$
forwards traffic from authenticated hosts to $F_1$ and $F_2$ (which
passes it through), and from guest hosts to $F_3$ (which performs the
required filtering of non-web traffic). However, some time later, the
network operator decides to transition to another configuration where
traffic from authenticated hosts is processed on $F_1$ and traffic
from guest hosts is processed on $F_2$ and $F_3$. Why might they want
to do this? Perhaps there is more traffic from guests hosts than from
authorized hosts, and the operator wishes to allocate more filtering
switches to guest traffic to better handle the load.

Implementing this update correctly turns out to be surprisingly
difficult. If we start by updating switches in an arbitrary order, we
can easily end up in a situation where the security policy is
violated. For example, if we update the ingress switch $I$ to the new
configuration without updating the filtering switches, then traffic
from guest hosts will be forwarded to $F_2$, which will incorrectly
pass it through to the internal network! One possible correct
implementation is to first update $I$ so it forwards traffic from
authenticated hosts to $F_1$, wait until all in-flight packets have
exited the network, update $F_2$ to filter non-web traffic, and
finally update $I$ again so that it forwards guest traffic to $F_2$ or
$F_3$. Finding this sequence is not impossible, but would pose a
significant challenge for the operator, who would have to reason about
all of the intermediate configurations, as well as their effect on
in-flight packets. By contrast, given encodings of the configurations
and the intended security policy, our system generates the correct
update sequence automatically.

\paragraph*{Ring.}
The second example involves the network topology shown in
Figure~\ref{fig:example}~(b). The network forwards packets around the
ring until they reach their destination. For example, if we forward
traffic clockwise around the ring, then a packet going from a host in
$N_1$ to a destination host in $N_3$ might be forwarded from $A$ to
$B$ to $C$. At all times, the network is required to be free of
forwarding loops---that is, no packet should arrive back at the same
port on a switch where it was previously processed. Initially the
network forwards packets around the ring in the clockwise direction,
as just described. Some time later, the network operator decides to
reverse the policy so that traffic goes around the ring in the
opposite direction. Implementing this update without introducing a
forwarding loop is challenging. In fact, if we implement updates at
the granularity of whole switch configurations, it is impossible! No
matter which switch we update first, the adjacent switch will forward
some packets back to it, thereby creating a loop. To implement the
update correctly, we must carefully separate out the traffic going to
each of networks $N_1$, $N_2$, and $N_3$, and transition those traffic
classes to the counter-clockwise configuration one by one. Assuming
the rules have this structure, our system generates the correct update
sequence automatically.

Note that the examples discussed in this section both depend on
updating individual rules on switches (rule granularity). However, the
formal model used in the rest of this paper, only considers updates to
whole switches (switch granularity). This is not a limitation: updates
at rule granularity can be easily reduced to switch granularity by
introducing an additional switch into the model for each rule. Our
tool assumes that this reduction has already been performed.

\section{Network Model}
\label{sec:model}

This section develops a simple abstract model of networks, and defines
the network update synthesis problem formally. Our model is based on
one proposed in previous work by Reitblatt et al.~\cite{frenetic-consistent-updates}. 


\paragraph*{Topologies and packets.}
A {\em network topology} is a tuple
$(\Switches,\Ports,\inport,\outport,\ingress)$, where $\Switches$ is a
finite set of switches; $\Ports$ is a finite set of ports with
distinguished ports $\DROP$ and $\WORLD$; $\ingress \in 2^\Ports$ is a
set of ingress ports; $\inport \in \Ports \times \Switches$ is a
relation such that for every port $\Port \in \Ports \setminus
\{\WORLD,\DROP\}$, there exists a unique switch $\Switch \in
\Switches$ with $\inport(\Port,\Switch)$; and $\outport \in \Switches
\times \Ports$ is a relation such that for every port $\Port \in
\Ports \setminus (\ingress \cup \{\WORLD,\DROP\})$, there exists a
unique switch $\Switch \in \Switches$ with
$\outport(\Switch,\Port)$. A \emph{packet} $\pt$ is a finite sequence
of bits. We assume that we can ``read off'' the values of standard
header fields such as Ethernet and IP addresses and TCP ports, and we
write $\PacketTypes$ for the set of all packets. A \emph{located
  packet} is a pair $(\Port,\pt)$, where \Port\ is a port and \pt\ is
a packet.

\paragraph*{Policies and updates.}
The switches in the network make decisions about how to forward
packets by examining their headers and the ingress ports on which they
arrive. We model this behavior using switch policies: a \emph{switch
  policy} \SwitchPolicy\ is a partial function \Partialt{\Ports \times
  \PacketTypes}{\Ports \times \PacketTypes}. A switch policy
\SwitchPolicy\ is {\em compatible} with a switch \Switch\ if whenever
\SwitchPolicy\ is defined on $(\Port,\pt)$ and returns $(\Port',\pt')$
then $\inport(\Port,\Switch)$ and $\outport(\Switch,\Port')$. Note
that real switches can forward packets out multiple ports. For
simplicity, in this paper, we restrict our attention to linear traces
of packets and only consider switch policies that generate at most one
packet.

A {\em network policy} \NetPolicy\ is a function
$\Totalt{\Switch}{\SwitchPolicy}$ where for all switches $\Switch \in
\Switches$, the switch policy \NetPolicy(\Switch) is compatible with
\Switch. The \emph{path} of a packet through the network is determined
by the topology and network policy. 

An {\em update} is a pair $(\Switch,\SwitchPolicy)$ consisting of a
switch \Switch\ and a switch policy \SwitchPolicy, such that
\SwitchPolicy\ is compatible with \Switch. Given a network policy
$\NetPolicy$ and an update $(\Switch,\SwitchPolicy)$, the expression
$\NetPolicy[\Switch \gets \SwitchPolicy]$ denotes a network policy
$\NetPolicy'$, where $\NetPolicy'(\Switch)=\SwitchPolicy$ and
$\NetPolicy'(\Switch')=\NetPolicy(\Switch')$ if $\Switch' \neq
\Switch$. Note that an update only modifies the policy for a single
switch.

\paragraph*{Commands and states.}
A {\em command} \command\ is either an update or the special command
\wait. A \wait{} command models the pause between updates needed to
ensure that packets that entered the network before the previous
command will leave the network before the next command. Intuitively,
waiting ``long enough'' makes sense only for network policies which
force every packet to leave in a bounded number of steps. This is
formalized below as the notion of \emph{wait correctness}.

A {\em network state} \NetState\ is a tuple
$(\LocPkt,\NetPolicy,\waitBool,\comSeq)$, where \LocPkt\ is a located
packet, \NetPolicy\ is a network policy, \waitBool\ is a Boolean
modeling whether updates are enabled, and \comSeq\ is a command
sequence. Note that our model only includes a single packet in the
network at any given time. As we are only interested in properties
involving paths of individual packets through the network,
intuitively, this is sufficient. However, we also need to be able to
generate new packets at ingress ports at any given time during the
execution of the network. One can prove (in a straightforward way)
that this model is equivalent to a full
model~\cite{frenetic-consistent-updates} with respect to LTL
properties of paths of individual packets.

\paragraph*{Traces.}
A {\em network transition} is a relation
$\Step{\NetState}{}{\NetState'}$ on states. There are four types of
transitions:
\begin{compactitem}
\item A {\em packet move} if 
\[
\begin{array}{rcl}
\NetState & = & ((\Port,\pt),\NetPolicy,\waitBool,\comSeq)\\
\NetState' &= & ((\Port',\pt'),\NetPolicy,\waitBool,\comSeq)
\end{array}
\]
and $\Port \not\in \{\WORLD,\DROP\}$, where there exists a switch
$\Switch$ such that $\inport(\Port,\Switch)$ and
\[
\NetPolicy(\Switch)(\Port,\pt)=(\Port',\pt')
\]
or $\pt=\pt'$ and $\Port=\Port'=\WORLD$ or $\pt=\pt'$ and
$\Port=\Port'=\DROP$.
\item An {\em update transition} if 
\[
\begin{array}{rcl}
\NetState &=& (\LocPkt,\NetPolicy,\false,(\Switch,\SwitchPolicy).\comSeq)\\
\NetState' &=& (\LocPkt,\NetPolicy[\Switch \gets \SwitchPolicy],\false,\comSeq)
\end{array}
\]
\item A {\em wait transition} if 
\[
\begin{array}{rcl}
\NetState & = & (\LocPkt,\NetPolicy,\waitBool,\wait.\comSeq)\\
\NetState' & = & (\LocPkt,\NetPolicy,\true,\comSeq) 
\end{array}
\]
The wait transition disables update transitions (by setting
$\waitBool$ to true), thus modeling the semantics of wait commands as
explained above.
\item A {\em new packet transition} if 
\[
\begin{array}{rcl}
\NetState &=& ((\Port,\pt),\NetPolicy,\waitBool,\comSeq)\\
\NetState' &=& ((\Port',\pt'),\NetPolicy,\false,\comSeq)
\end{array}
\]
where $\Port' \in \ingress$. (Note that there is no condition on the
new packet.)  This transition models that, non-deterministically, we
can decide to track a new packet. 
\end{compactitem}

A {\em network trace} $\NetTrace$ is an infinite sequence of states
$\NetState_0\NetState_1 \ldots$ such that for all $i \geq 0$ we have
that $\Step{\NetState_i}{}{\NetState_{i+1}}$.  A network trace {\em
  initialized with} a policy $\NetPolicy$ and a command sequence
$\comSeq$ is a network trace such that $\NetState_0 =
(\LocPkt,\NetPolicy,b_w,\comSeq)$, for some located packet $\LocPkt$
and Boolean $b_w$.


A {\em one-packet trace} $t = \LocPkt_0 \LocPkt_1 \ldots $ is a
sequence of located packets that conforms to the network
topology. That is, for all $i < |t|$, we have that if $\LocPkt_i =
(\Port,\pt)$ and $\LocPkt_{i+1} = (\Port',\pt')$, then there exists a
switch $\Switch \in Switches$, such that $\inport(\Port,\Switch)$ and
$\outport(\Switch,\Port')$.  A {\em complete one-packet trace} is a
finite trace such that $\LocPkt_0 \LocPkt_1 \ldots \LocPkt_n$ such that
$\LocPkt_0 = (\Port,\pt)$ where $\Port$ is in $\ingress$ and
$\LocPkt_n = (\WORLD,\pt')$ or $\LocPkt_n = (\DROP,\pt')$.

A one-packet trace $t = \LocPkt_0 \LocPkt_1 \ldots
\LocPkt_n$ is {\em contained} in a network trace $\NetTrace =
\NetState_0 \NetState_1 \ldots$ if there is a function $f$ (witnessing
the containment) from $[0,n]$ to $\mathbb{N}$ with the following
properties:
\begin{compactitem}
\item for all $i \in [0,n-1]$, $f(i) < f(i+1)$;
\item for all $i \in [0,n]$, we have that if $\NetState_{f(i)} =
  (\LocPkt,\NetPolicy,\waitBool,\comSeq)$, then $\LocPkt = \LocPkt_i$;
\item for all $i \in [0,n-1]$, the transitions occurring between
  $f(i)$ and $f(i+1)-1$ in $\NetTrace$ are only update and wait
  transitions, and the transition between $f(i+1)-1$ and $f(i+1)$ is
  a packet move transition.
\end{compactitem}
A given network trace may contain traces of many packets generated by
new packet transitions.

\paragraph*{Wait and command correctness.}
A network state $\NetState$ is {\em wait-correct} if, intuitively, the
packet cannot stay in the network for an unbounded amount of time.
Formally, $\NetState$ is wait-correct if for all infinite network
traces $\NetTrace = \NetState_0\NetState_1 \ldots$ such that
$\NetState_0 = \NetState$, and the transition from $\NetState_0$ to
$\NetState_1$ is a wait transition, either there exists $i \in
\mathbb{N}$ such that for all $j \in \mathbb{N}$ with $j > i$ the
packet at $\NetState_j$ is located at $\DROP$ or at $\WORLD$, or there
exists $i \in \mathbb{N}$ such that the transition from $\NetState_i$
to $\NetState_{i+1}$ is a new packet transition.


The function $\infin(\cpt)$ appends an infinite suffix of the form
$\LocPkt_n^\omega$ to the complete one-packet trace $\cpt =
\LocPkt_0\LocPkt_1 \ldots \LocPkt_n$. Recall that a complete
one-packet trace ends with the packet located at $\DROP$ or $\WORLD$,
so $\infin(\cpt)$ models a packet staying outside of the network.

\paragraph*{LTL.}
We now define LTL formulas and their semantics over infinite
one-packet traces. Atomic formulas are of the form $\mathit{packet =
  \pt}$ or $\mathit{port = \Port}$. A formula $\varphi$ is an LTL
formula, if it is an atomic formula, or is of the form $\neg
\varphi_1$, $\varphi_1 \vee \varphi_2$, $X \varphi$, $\varphi_1 U
\varphi_2$, where $\varphi_1$ and $\varphi_2$ are LTL formulas.  As is
standard, we will also use connectives $F$ and $G$ that can be defined
in terms of the other connectives.  Let $t$ be an infinite one-packet
trace $\LocPkt_0 \LocPkt_1 \ldots$.  We have that $t \models
\mathit{packet} = \pt$ if there exists a port $\Port$ such that
$\LocPkt_0 = (\pt,\Port)$. Similarly, we have that $t \models
\mathit{port} = \Port$ if there exists a packet $\pt$ such that
$\LocPkt_0 = (\pt,\Port)$. The semantics of Boolean and temporal
connectives is standard. An example of an LTL specification for the
distributed firewall example is given in Figure~\ref{fig:nusmv} in
Section~\ref{sec:implementation}.

Let $\varphi$ be an LTL formula. A network trace $\NetTrace$ satisfies
an LTL formula $\varphi$ (written $\NetTrace \models \varphi$) if
for all complete one-packet traces $\cpt$ contained in $\NetTrace$, we
have that $\infin(\cpt) \models \varphi$. Let $\comSeq =
\command_0\command_1 \ldots \command_{k-1}$ be a sequence of commands,
and let $\NetPolicy$ be a network policy. A sequence of network
policies $\NetPolicy_0 \NetPolicy_1 \ldots \NetPolicy_n$ is induced by
$\comSeq$ and $\NetPolicy$, if
\begin{compactitem}
\item $\NetPolicy_0 =\NetPolicy$
\item for all $i$ in $[0,k-1]$, if $\command_i = \wait$ then
  $\NetPolicy_{i+1} =\NetPolicy_i$ 
\item for all $i$ in $[0,k-1]$, if $\command_i = (\Switch,\SwitchPolicy)$ then
  $\NetPolicy_{i+1} = \NetPolicy_i[\Switch \gets \SwitchPolicy]$ 
\end{compactitem}
We write $\Step{\NetPolicy}{\comSeq}{\NetPolicy'}$ if the last element
of the sequence induced by $\comSeq$ is $\NetPolicy'$. A command
sequence $\comSeq$ is {\em correct} with respect to a formula
$\varphi$ and policy $\NetPolicy_i$ if for all network traces
$\NetTrace$ initialized with $\NetPolicy_i$ and $\comSeq$, we have
that $\NetTrace$ is wait-correct and $\NetTrace \models \varphi$.

\paragraph*{Update synthesis problem.} 
With this notation in hand, we are now ready to formally state the
\emph{network update synthesis problem}. Given an initial network
policy $\NetPolicy_i$, a final network policy $\NetPolicy_f$, and a
specification $\varphi$, construct a sequence of commands $\comSeq$
such that:
\begin{compactitem}
\item $\Step{\NetPolicy_i}{\comSeq}{\NetPolicy_f}$, and 
\item $\comSeq$ is correct with respect to $\varphi$ and $\NetPolicy_i$.
\end{compactitem}
The next section develops an algorithm that solves this problem.

\renewcommand{\thealgorithm}{}
\begin{figure}[t!]
\begin{algorithmic}[1]

\Statex{\hspace*{-6.5mm}\textbf{Procedure}~\textsc{OrderUpdate}($\NetPolicy_i,\NetPolicy_f,\varphi$)}
\Require~Initial network policy $\NetPolicy_i$, final network policy $\NetPolicy_f$, and LTL specification $\varphi$.  
\Ensure~Simple and careful sequence of switch updates $L$, if it exists
\If {hasLoops($\NetPolicy_i$) $\vee$ hasLoops($\NetPolicy_f$)}
\State \Return ``Loops in initial or final configuration.''
\Else
  \State $W \gets$ false\Comment{Wrong configurations.}   
  \State $V \gets$ false\Comment{Visited configurations.} 
  \State (ok,~$L$) $\gets$ DFSforOrder($\NetPolicy_i$,~$\bot$)
  \If {ok} 
     \State \Return $L$
  \Else 
     \State \Return ``No simple and careful update sequence exists.''
  \EndIf
\EndIf
\Statex
\Statex{\hspace*{-6.5mm}\textbf{Procedure}~\textsc{DFSforOrder}($\NetPolicy,cs$)}
\Require~Current network policy $\NetPolicy$, most recently updated switch $cs$.
\Ensure~Boolean ok if a correct update sequence exists; $L$ correct sequence 
of switch updates
   \If {$\NetPolicy = \NetPolicy_f$} 
      \State \Return (true,$[\NetPolicy]$) \Comment{Reached final configuration.}   
   \EndIf
   \If {$\NetPolicy \models V$} 
      \State \Return (false,$[]$)\Comment{Already visited $\NetPolicy$.}      
   \EndIf

   \State $V \gets V \vee \NetPolicy$\Comment{Add to visited configurations.} 
   
   \If {$\NetPolicy \models W$} 
     \State \Return (false,$[]$)\Comment{Previous counterexample applies.}  
   \EndIf

   \If {$cs \neq \bot$}\Comment{If there was a previous update,}
     \State\label{line:checkNewLoops} (ok,cex) $\gets$
     hasNewLoops($\NetPolicy$,$cs$)\Comment{Check for forwarding loops.}
     \If {($\neg$ ok)}
       \State $W \gets W \vee$ analyzeCex(cex)\Comment{Learn from loop counterexample.} 
       \State \Return (false,$[]$)
     \EndIf
   \EndIf

   \State\label{line:modelcheck} (ok,cex) $\gets$ ModelCheck($\NetPolicy$,$\varphi$)
   \If {($\neg$ ok)}
     \State $W \gets W \vee$ analyzeCex(cex)\Comment{Learn from property counterexample.} 
     \State \Return (false,$[]$)
   \EndIf

   \ForAll {$(\NetPolicy_\next,cs) \in$ NextPolicies$(\NetPolicy)$ }\Comment{Try to update one more switch.} 
     \State (ok,$L$) $\gets$ DFSforOrder($\NetPolicy_\next,cs$)\Comment{Recursive call.}
     \If {ok}
       \State\label{line:addcommand} \Return (true,$\NetPolicy::\wait::L$)
     \EndIf
  \EndFor 
  \State \Return (false,$[]$)
\Statex
\Statex
\end{algorithmic}
\caption{\textsc{OrderUpdate} Algorithm.}
\label{algo:order}
\end{figure}

\section{Update Synthesis Algorithm}
\label{sec:algorithms}

This section presents an algorithm that synthesizes correct network
updates automatically. The algorithm attempts to find a sequence of
individual switch updates that transition the network from the initial
configuration to the final configuration, while ensuring that the path
of every packet traversing the satisfies the invariants specified by
the operator.\footnote{We assume that the topology is fixed, so that a
  network configuration is just a network policy.}  It works by
searching through the space of possible update sequences, but
incorporates three important optimizations aimed at making synthesis
more efficient.

\paragraph*{Optimizations.}
The first optimization restricts the search space to solutions that
update every switch in the network at most once. We call solutions
with this property \emph{simple}. Because the space of simple 
solutions is much smaller than the full space of solutions, this leads
to a much more efficient synthesis procedure in practice.

The second optimization restricts the search space to solutions for
which the synthesis procedure can efficiently check correctness.
Because the network continues to process packets even as it is being
updated, in general a packet may traverse the network during multiple
updates. Hence, to ensure the correctness of the path of such a
packet, it is necessary to check properties of sequences of network
configurations, which can lead to an explosion of model checking
tasks. We therefore introduce the notion of {\em careful}
updates---update sequences where the system pauses between each step
to ensure that all packets that were in flight before the step will
have exited the network. There is one caveat worth noting: waiting
only makes sense only for configurations for which every packet leaves
a network after a bounded number of steps. To ensure this is possible,
we require configurations to be loop-free in the sense that the policy
has the property that every packet is processed by a given switch at
most once. We thus have that every packet is in the network during at
most one update. For such loop-free configurations, every packet
either has a path using the configuration before a given step was
applied, or the configuration after the step was applied. This enables
us to check correctness of configurations separately. We do not need
to check all possible configurations, we only need to check those
encountered during the search.

The third optimization uses counterexamples to reduce the number of
calls to the model checking procedure. The purpose of a call to the
model checker is to check that all possible packet paths satisfy the
specified invariants. However, if the model checker identifies a path
that does not satisfy an invariant, the path is returned as a
counterexample and can be used to eliminate future configurations
quickly. In particular, any intermediate configurations in which the
switches are configured in the same way as in the counterexample can
be eliminated without having to consult the model checker.

\paragraph*{Algorithm.}
Figure~\ref{algo:order} presents pseudocode for the
\textsc{OrderUpdate} algorithm. It returns a sequence of careful and
simple commands that implement the update correctly, or fails if no
such sequence exists. The notions of simple and careful command
sequences are defined formally below. We speak of sequences of
commands, rather than sequences of updates, because we also include
wait-commands for the reasons described above. The rest of this
section describes the algorithm in detail and proves that it is sound
and complete with respect to simple and careful command
sequences. That is, if a simple and careful command sequence exists,
then the algorithm will find it.

As we are interested only in simple command sequences, the main task
is to find an order of switch updates. To do this, it uses a
depth-first search, where at each recursive call, we update one
switch. We consider only switches whose policy is different in the
initial and final configurations. We opted for depth-first search as
we expect that, in common cases, many update sequences will lead to a
solution.

Before starting the search, we check that the initial and final
configurations have no loops---otherwise a simple and careful sequence
of commands does not exist. This is done by two calls to a function
$\mathrm{hasLoops}$. During the search, we check that
each update that we encounter has not introduced new loops into the
configuration.  This is done in the auxiliary function
$\mathrm{hasNewLoops}(\NetPolicy_\next,cs)$, which takes as parameters
the updated network policy and the switch that was updated. This check
can be easily implemented using an LTL formula, as any new loop must
pass through the updated switch.


The search maintains a formula $V$ that encodes visited
configurations, and a formula $W$ that encodes the set of
configurations excluded by counterexamples so far.  The auxiliary
function $\mathrm{analyzeCex}$ analyzes a counterexample, and outputs
a formula representing the set of switches that occurred in the
counterexample, and whether these switches were already updated.

If the current configuration was not visited before, and is not
eliminated by previous counterexamples, we check whether all packet
traces traversing this configuration satisfy the LTL specification
$\varphi$. This is the purpose of the call to
$\mathrm{ModelCheck}(\NetPolicy,\varphi)$. In our implementation, we
use NuSMV~\cite{CCGGPRST02} as a back-end model checker.

If the current configuration passes all these tests, we continue the
depth-first search, with next configurations being those where one
more switch is updated. If we reach the final configuration, we pop
out of the recursive calls, and prepend the corresponding updates
(separated by wait commands) to the command sequence returned.

\paragraph*{Soundness.}
Now we prove that \textsc{OrderUpdate} is sound. For the remainder of this
section, let us fix a specific network topology
$(\Switches,\Ports,\inport,\outport,\ingress)$.


A network policy $\NetPolicy$ satisfies an LTL formula $\varphi$
(denoted by $\NetPolicy \models \varphi$), if for all network traces
$\NetTrace$ initialized with $\NetPolicy$ and the empty command
sequence, we have that $\NetTrace \models \varphi$.  A network policy
$\NetPolicy$ {\em induces} a one-packet trace $t$, if there exists a
network trace $\NetTrace$ initialized by $\NetPolicy$ and the empty
sequence of commands, such that $\NetTrace$ contains $t$.

A policy $\NetPolicy$ is {\em loop-free} if, intuitively, there is no
loop in the graph given by the network topology and $\NetPolicy$.
More formally, for all sequences $w = \Port_0\Switch_0\Port_1\Switch_1
\ldots \Port_k\Switch_k$ that conform to the network topology and to
$\NetPolicy$, we have that no port (and no switch) occurs more than
once in $w$.  A sequence $\Port_0\Switch_0\Port_1\Switch_1 \ldots
\Port_k\Switch_k$ conforms to the network topology and to
$\NetPolicy$, if for all $i\in [0,k-1]$, we have that
$\inport(\Port_i,\Switch_i)$, and there exist packets $\pt$ and $\pt'$
such that $\NetPolicy(\Switch_i)(\Port_i,\pt)= (\Port_{i+1},\pt')$.

Let $\comSeq = \command_0\command_1 \ldots \command_{n-1}$ be a
command sequence, and let $\NetPolicy$ be a network policy. Let
$\NetPolicy_0 \NetPolicy_1 \ldots \NetPolicy_n$ be the sequence of
network policies induced by $\comSeq$ and $\NetPolicy$.  The command
sequence $\comSeq$ is {\em careful} with respect to an LTL formula
$\varphi$ and a network policy $\NetPolicy$ if
\begin{compactitem}
\item for all $i \in [0,n-1]$, if $i$ is odd, then $\command_i = \wait$, 
\item for all $i \in [0,n]$, $\NetPolicy_i$ is loop-free, and
\item for all $i \in [0,n]$, $\NetPolicy_i \models \varphi$. 
\end{compactitem}  

Let $\comSeq$ be a careful sequence of commands. Let $\NetPolicy_i$ be
a network policy.  Let $\NetTrace = \NetState_0 \NetState_1 \ldots$ be
a network trace initialized with $\NetPolicy_i$ and $\comSeq$. Let $t
= \LocPkt_0\LocPkt_1 \ldots \LocPkt_n$ be a one-packet trace contained
in $\NetTrace$.  Let $\sigma = \Switch_0\Switch_1 \ldots
\Switch_{n-1}$ be a sequence of switches such that $i < n$, we have
that if $\LocPkt_i = (\Port,\pt)$ and $\LocPkt_{i+1} = (\Port',\pt')$,
then $\inport(\Port,\Switch_i)$ and $\outport(\Switch_i,\Port')$. Our
first lemma states that the path of every packet is affected by at
most one update.


\begin{lemma}\label{lemma:eachpacketoneupdate}
Let $f$ be the function witnessing the containment of $t$ in
$\NetTrace$. There is at most one update transition in $\NetTrace$
between $f(0)$ and $f(n)$.
\end{lemma}
\begin{proof}
We use the fact that $\comSeq$ is careful, specifically that every
command at an odd position in the sequence of commands is $\wait$.
Let us assume that there are two update transitions between $f(0)$ and
$f(n)$ in $\NetTrace$. Let these two update transitions occur at
network states $\NetState_i$ and $\NetState_j$ such that $f(0) \leq i
< j < f(n)$. As $\comSeq$ is careful, we have that there is a
wait-transition that occurs at $\NetState_w$, where $i < w < j$.  As
wait-transitions disable updates (by setting $\waitBool$ to true; this
is because the wait command models waiting long enough so that packets
that entered the network before the previous update will leave the
network before the next update), there has to be a new packet
transition $\NetState_p$, where $w < p < j$. This contradicts the fact
that $t$ is contained in $\NetTrace$, which concludes the proof.
\end{proof}

The second lemma states that a path of every packet in the network
could have occurred in of the intermediate configurations. That is, no
packet takes a path non-existent in any of the configurations, even
though the packet might be in-flight during the updates.

\begin{lemma}\label{lemma:eachpacketseesoneconfig}
There exists $i \in \mathbb{N}$ such that $\NetState_i$ induces $t$.
\end{lemma}
\begin{proof}
We use the fact that $\comSeq$ is careful, specifically that $\comSeq$
is such that each $\NetPolicy_i$ (for $0 \leq i \leq n)$ is loop-free.

Let $f$ be the function witnessing the containment of $t$ in
$\NetTrace$. By Lemma~\ref{lemma:eachpacketoneupdate}, we have that
there is at most one update transition in $\NetTrace$ between $f(0)$
and $f(n)$. Let the update transition be given by the update
$(\Switch,\SwitchPolicy)$. Let $\NetState_j$ ($f(0) \leq j < f(n)$) be
the network state in which the update transition occurs.  We show that
either $\NetState_j$ or $\NetState_{j+1}$ induces $t$.

Now let us consider $\sigma$ (defined above), which intuitively is the
sequence of switches that a packet sees as it traverses the
network. We analyze the following cases:
\begin{compactitem}
\item $\Switch$ does not occur in $\sigma$. Then $t$ was not influenced
  by the update, and is thus induced by $\NetState_j$.
\item $\Switch$ occurs in $\sigma$, but only once. Let $l$ be the
  smallest position in $t$ such that there exist a port $\Port$ and a
  packet $\pt$ such that $\LocPkt_l = (\Port,\pt)$ and
  $\outport(\Switch,\Port)$.  If $f(l)$ is less than $j$ (i.e. the
  packet was at $\Switch$ before the update happened) then $t$ is
  induced by $\NetPolicy_j$. If $f(l)$ is greater than $j$, then $t$
  is induced by $\NetPolicy_{j+1}$.
\item $\Switch$ occurs more than once in $\sigma$. Let $\Switch_k$
  $\Switch_{k+1}$ $\Switch_l$ be the subsequence of switches between
  two closest occurrences of $\Switch$ in $\Sigma$.  As none of the
  switches in the subsequence were updated, we have that
  $\NetPolicy_j$ or $\NetPolicy_{j+1}$ is not loop-free, which
  contradicts the assumption that $\comSeq$ is careful.
\end{compactitem}
This completes the proof.
\end{proof}

\noindent The third lemma states that carefulness (which is easily
checkable) implies correctness.

\begin{lemma}\label{lemma:careful}
If a command sequence $\comSeq$ is careful with respect to an LTL
formula $\varphi$ and a network policy $\NetPolicy$, then $\comSeq$ is
correct with respect to $\varphi$ and $\NetPolicy$.
\end{lemma}
\begin{proof}
To show that $\comSeq$ is correct with respect to $\varphi$ and
$\NetPolicy$, we need to show that for all network traces $\NetTrace$
initialized with $\NetPolicy_i$ and $\comSeq$, we have that
$\NetTrace$ is wait-correct and $\NetTrace \models \varphi$. We first
prove that $\NetTrace$ is wait-correct.  Let $\NetTrace=\NetState_0
\NetState_1 \ldots$ be a network trace, and let $i$ be such that the
transition from $\NetState_i$ to $\NetState_i$ is a
wait-transition. We need to prove that for all infinite network traces
that start at $\NetState_i$, and which do not contain a new packet
transition, we have that the packet ends at the port $\DROP$ or
$\WORLD$ after a finite number of steps. Consider a network trace
$\NetTrace'$ that starts at $\NetState_i$ and does not contain a new
packet transition. Let us consider the unique one-packet trace $t$
that starts at the last new packet transition before $\NetState_i$ in
$\NetTrace$ (or which starts at the first position of $\NetTrace$ if
there is no new packet transition in $\NetTrace$), and continues as in
$\NetTrace'$. Consider a prefix $t'$ of $t$ longer than the number of
switches and ports in the network. By the proof of
Lemma~\ref{lemma:eachpacketseesoneconfig}, $t'$ is induced by
$\NetState_p$, for $p$ such that $p < i$. As $\comSeq$ is careful, we
can conclude that $t'$ is induced by a network state with a loop-free
network policy, which means that the packet reaches $\DROP$ or
$\WORLD$ after a finite number of steps.

We now prove that $\NetTrace \models \varphi$.  Let
$\NetTrace=\NetState_0 \NetState_1 \ldots$ be a network trace and
$\cpt = \LocPkt_0\LocPkt_1 \ldots \LocPkt_n$ be a complete one-packet
trace contained in $\NetTrace$.  We show that $\infin(\cpt)) \models
\varphi$.  By Lemma~\ref{lemma:eachpacketseesoneconfig}, we have that
there exists $i \in \mathbb{N}$ such that $\cpt$ is induced by a
$\NetState_i$. As $\comSeq$ is careful, we have that for all complete
one-packet traces $t'$ induced by $\NetState_i$, we have that $t'
\models \varphi$. Therefore, we can conclude that $\cpt \models
\varphi$, and as there were no conditions on how $\cpt$ was chosen, we
have that $\NetTrace \models \varphi$. This concludes the proof.
\end{proof}

\begin{theorem}[Soundness]
Given an initial policy $\NetPolicy_i$ a final policy $\NetPolicy_f$,
and an LTL formula $\varphi$, \textsc{OrderUpdate} returns a command sequence
$\comSeq$, then $\Step{\NetPolicy_i}{\comSeq}{\NetPolicy_f}$, and
$\comSeq$ is correct with respect to $\varphi$ and $\NetPolicy_i$.
\end{theorem}
\begin{proof}
It is easy to show that if \textsc{OrderUpdate} returns $\comSeq$, then
$\Step{\NetPolicy_i}{\comSeq}{\NetPolicy_f}$. Each update in the
returned sequence changes a switch policy of one switch $\Switch$ to
the policy $\NetPolicy_f(\Switch)$, and the algorithm terminates when
all switches $\Switch$ such that $\NetPolicy_i(\Switch)
\neq \NetPolicy_f(\Switch)$ have been updated.
Let $\NetPolicy_0 \NetPolicy_1 \ldots \NetPolicy_n$ be induced by
$\comSeq$ and $\NetPolicy_i$.  
We show that if \textsc{OrderUpdate} returns $\comSeq$, then $\comSeq$ is
careful with respect to $\varphi$ and $\NetPolicy_i$. 
To prove that
$\comSeq = \command_0\command_1 \ldots \command_{n-1}$ is careful, we
show that:
\begin{compactitem}
\item for all $j \in [0,n-1]$, if $j$ is odd, then $\command_j =
  \wait$. One can simply observe that this is true, given how the
  sequence of updates is constructed in the algorithm
  (Line~\ref{line:addcommand}).
\item for all $j \in [0,n]$, $\NetPolicy_j$ is loop-free. This holds,
  as we check that the initial configuration is loop-free, and that
  each update does not introduce a loop (Line~\ref{line:checkNewLoops}). 
\item for all $j \in [0,n]$, $\NetPolicy_j \models \varphi$. This is
  ensured by the call to a model checker (Line~\ref{line:modelcheck}).
\end{compactitem}  

Finally we can use Lemma~\ref{lemma:careful} to infer that $\comSeq$
is careful with respect to $\varphi$ and $\NetPolicy_i$.
\end{proof}

\paragraph*{Completeness.}
The \textsc{OrderUpdate} algorithm is also complete with respect to
simple and careful command sequences. Let $\comSeq =
\command_0\command_1 \ldots \command_{n-1}$ be a command
sequence. Such a sequence is {\em simple} if for each $\Switch \in
\Switches$ there exists at most one $i$ in $[0,n]$ such that
$\command_i$ is an update of the form $\Switch,\SwitchPolicy$. The
proof (omitted here) uses the fact that \textsc{OrderUpdate} searches
through all such sequences (more precisely, all such sequences that do
not use multiple $\wait$ commands in a row). The following proposition
characterizes the cases where the algorithm returns a solution.

\begin{proposition}
\label{prop:complete}
Given an initial network policy $\NetPolicy_i$, a final network policy
$\NetPolicy_f$, and a specification $\varphi$, if there exists a
simple and careful sequence of commands $\comSeq$ such that
$\Step{\NetPolicy_i}{\comSeq}{\NetPolicy_f}$, then
\textsc{OrderUpdate} returns one such sequence.
\end{proposition}


\section{Implementation}
\label{sec:implementation}



We have built an implementation of \textsc{OrderUpdate} in OCaml. The
functions $\mathrm{ModelCheck}(\NetPolicy,\varphi)$ and
$\mathrm{hasNewLoops}(\NetPolicy,cs)$ are implemented by calling out
to the NuSMV~\cite{CCGGPRST02} model checker on suitable encodings of
the network configuration. More specifically, the function
$\mathrm{hasNewLoops}(\NetPolicy_\next,cs)$ takes as parameters the
updated network policy and the switch that was updated, and checks
that no new loops were introduced by the update. This check can be
performed using the LTL formula $G (cs \rightarrow \neg X (F\ cs))$,
as any newly introduced loops must pass through the updated switch.

\paragraph*{NuSMV models.} 
The NuSMV encodings of network configurations are similar to the
formal model described in Section~\ref{sec:model}: Packets are
represented as tuples consisting of {\tt src}, {\tt dst}, and {\tt
  purpose}, where {\tt src} is source of the packet (e.g., a ``guest''
host), {\tt dst} is the destination of the packet, and {\tt purpose}
is a general field (e.g. ``Web traffic''). Switch policies are encoded
as NuSMV expressions over these variables ({\tt src}, {\tt dst}, and
{\tt purpose}) as well as ingress ports. The model has a single entry
point---a port {\tt Start} from which a packet can transition to an
ingress port on any switch. Finally, as in Section~\ref{sec:model}, we
reduce the size of the NuSMV input by transitioning located packets to
the next ingress port after forwarding---i.e., we inline the links
between the output port on one switch and the ingress port at
another. Figure~\ref{fig:nusmv} gives the NuSMV encoding of the
initial configuration for the firewall example from
Section~\ref{sec:overview}.


\paragraph*{Rule granularity.} 
Recall that we represent switch policies as partial functions, and we
model updates that apply at the granularity of whole switches. Of
course, in real switches, policies are represented using {\em rules}
that ``match'' the domain of the function, and the switch forwards
packets according to the best matching rule. Hence, it is important to
be able to encode finer-grained updates that only modify particular
rules on switches---indeed, such updates are used in both of the
motivating examples from Section~\ref{sec:overview}. Fortunately, rule
granularity can be easily reduced to switch granularity: we transform
the switch into a sequence of switches, where each switch forwards
packets matched by one rule, and passes all unmatched packets along to
the next switch. We use this technique in many of our examples.

\begin{figure}[t]
  \centering
  \begin{alltt}\small
MODULE main
VAR
    port : \{I_0, F1_0, F2_0, F3_0, START, WORLD, DROP\};
    src : \{Auth, Guest\};
    purpose : \{Web, Other\};
ASSIGN
    next(port) := case
        port = START : I_0;
        port = I_0 \& src = Auth : \{F1_0, F2_0\};
        port = I_0 \& src = Guest : F3_0;
        port = F1_0 : {WORLD};
        port = F2_0 : {WORLD};
        port = F3_0 \& purpose = Web : WORLD;
        port = F3_0 \& purpose = Other : DROP;
        port = WORLD : WORLD;
        port = DROP : DROP;
    esac;
    next(src) := src;
    next(purpose) := purpose;
INIT port = START;
LTLSPEC G (purpose = Other \& src = Guest -> F port = DROP) \& 
          ((src = Auth | src = Guest \& purpose = Web) -> F port = WORLD);
  \end{alltt}
  \caption{NuSMV encoding of firewall example.}
\label{fig:nusmv}
\end{figure}

\paragraph*{Other algorithms.}

Besides the \textsc{OrderUpdate} algorithm, we have also implemented
two additional algorithms for comparison purposes.
The \textsc{Refine} algorithm provides a direct implementation of a
counterexample-guided synthesis approach to our problem. In this
approach, we add a Boolean variable for each switch to model whether
the switch has updated or not. We allow switches to update as the
packet traverses the network, with no more than one switch updating
per new packet transition. We use counterexamples learned from NuSMV
to refine our model, explicitly preventing the update order appearing
in the counterexample. The process continues until either the final
configuration cannot be reached or any sequence of updates possible in
the refined model is safe.

The \textsc{ConfigPairs} algorithm has the same structure as the
\textsc{OrderUpdate} algorithm, but includes an additional Boolean
variable for the switch being updated. This variable models whether
the switch has updated or not. We allow the switch to update at any
time, including while the packet traverses the network. In effect,
there is a model checking call for each pair of configurations in the
worst case (as opposed to a call per configuration). This is because
the algorithm in the preceding section relies on
Lemmas~\ref{lemma:eachpacketoneupdate}
and~\ref{lemma:eachpacketseesoneconfig}, rather than on checking pairs
of configurations.

\section{Experiments}
\label{sec:experiments}

To evaluate the effectiveness of our implementation, we used it to
generate update sequences for several examples. To provide a
comparison, we compared our main \textsc{OrderUpdate} algorithm to our
own implementations of the (simpler) \textsc{Refine} and
\textsc{ConfigPairs} algorithms.





\paragraph*{Goals.}
The most important parameters of the network update problems are $N$,
the total number of switches in the network, and $M$, the number of
switches whose switch policy differs between initial and final
configuration. Note that the size of the solution space is $M!$. The
goal of our experimental evaluation is to quantify how our tool scales
with growing $M$ and $N$, both for problems where a solution exists
and for problems where the solution does not exist.  We believe that
an important class of network update problems that occurs in practice
is when $N$ is on the order of $1000$, and $M$ is on the order of
$10$---such updates arise when there is a problem on a small number of
nodes and the network must route around it.

\paragraph*{Benchmarks.}
We ran our tests on specific network configurations, parameterized
by $N$ and $M$. The topology of the network, depicted in
Figure~\ref{fig:experiments}~(a), is as follows: the network has an
inner part consisting of a sparse but connected graph, and an outer
part with a larger number of nodes and ingresses reachable in two
hops. In the experiments, we removed several of the switches in the
inner part of the network while maintaining connectivity, so that at
all times each ingress port is reachable from the other
two. Intuitively, this experiment could model taking down switches for
maintenance. The two policies are computed using shortest-path
computations before and after the switches are removed. This
experiment allows us to both scale the inner part, increasing the
number of switches that differ between the policies, and also scale
the total number of switches by increasing the number of switches in
the outer parts.


\paragraph*{Results.}
We ran our experiments using a laptop machine with a 2.2 GHz Intel
processor and 4 GB RAM. We used NuSMV version 2.5.4 as the external
model checker. 

\begin{description}
\item[Scaling network size:] The first experiment
  tests how our tool scales with $N$ (the total number of nodes). We
  fixed the number of nodes updating at $13$ and ran the tool on
  graphs of size $100$, $250$, $500$, and $1000$. We ran each
  experiment using the \textsc{OrderUpdate} algorithm discussed in
  Section~\ref{sec:algorithms}, as well as \textsc{Refine} and
  \textsc{ConfigPairs} algorithms described in
  Section~\ref{sec:implementation}.  The \textsc{Refine}
  implementation failed on the two larger inputs. The results are
  reported in Figure~\ref{fig:experiments}~(b).

\item[Scaling update size:] The next experiment tests how our tool
  scales with $M$ (the number of nodes updating). In this experiment,
  we held $N$ (the total number of nodes) fixed at $500$ and ran the
  tool with the total number of nodes updating between $5$ and $15$.
  We show the results for \textsc{OrderUpdate} algorithm only, as the
  above experiments show that the other two do not perform well with
  $500$ nodes. The results are reported in
  Figure~\ref{fig:experiments}~(c).

\item[Impossible updates:] The final experiment tests how our tool
  performs on impossible updates---i.e., updates for which no safe and
  careful sequence of switch updates exists. We modified the benchmark
  slightly so that in the final configuration, the ingress switches
  drop packets destined for them instead of forwarding them out to the
  world. In this experiment, we used updates that affected $8$ of the
  nodes. The results of this experiment are shown in
  Figure~\ref{fig:experiments}~(d). We also report how the tool
  performs without counterexample analysis here (and not in the
  previous tables), as counterexamples are most helpful when there are
  many incorrect configurations. It is interesting to note that
  although \textsc{Refine} does not scale as well to large numbers of
  nodes, it is able to quickly determine when an update is impossible.
\end{description}

\paragraph*{Summary.}
Overall, our experiments show that our tool scales to the class of
network updates problems outlined above. For a network with $N=1000$
nodes, $M=13$ of which need to be updated, the running time is $18$
minutes. Our tool also scales for a larger number of nodes
updating. For $500$ nodes total, and $30$ nodes updating, the running
time is $10$ minutes. These running times are far too large for online
use by network operators, but we emphasize that we report on a
prototype tool---our primary goal was to confirm feasibility of our
approach. We leave building a well-engineered tool to future work.  We
note that if it is not possible to find an update, the tool takes much
longer to complete. This is because the tool needs to go through a
large number of possible update sequences. Here, our counterexample
analysis helps significantly, reducing the running time for the case
$N=500$, $M=8$ by $85\%$.  However, the tool does not scale well with
$M$ in impossible updates; with $M=10$ this example ran for over $45$
minutes.

\begin{figure}[t]
  (a)\\
  \centerline{%
  \begin{tikzpicture}[scale=.275,auto=left,
      every node/.style={circle,draw=black,minimum size=10}]
    \node (ingress1) at (-8,0) {in 1};
    \node (ingress2) at (6.47,4.24) {in 2};
    \node (ingress3) at (6.47,-4.24) {in 3};
    \node (i1) at (-3,0) {};
    \node (i2) at (-0.93,2.85) {};
    \node (i3) at (-0.93,-2.85) {$x$};
    \node (i4) at (2.43,1.76)  {};
    \node (i5) at (2.43,-1.76) {};
    \node (f11) at (-5.25,0) {};
    \node (f12) at (-5.25,2) {};
    \node (f13) at (-5.25,-2) {};
    \node (f21) at (4.25,3.09) {};
    \node (f22) at (3.07,4.70) {};
    \node (f23) at (5.42,1.47) {};
    \node (f31) at (4.25,-3.09) {};
    \node (f32) at (3.07,-4.70) {};
    \node (f33) at (5.42,-1.47) {};
    \path[->,every node/.style={font=\scriptsize}] (i1)
    edge [-] (i2) (i2)
    edge [-] (i4) (i4)
    edge [-] (i5) (i5)
    edge [-] (i3) (i3)
    edge [-] (i1) (ingress1)
    edge [-] (f11) (ingress1)
    edge [-] (f12) (ingress1)
    edge [-] (f13) (f11)
    edge [-] (i1) (f12)
    edge [-] (i1) (f13)
    edge [-] (i1) (ingress2)
    edge [-] (f21) (ingress2)
    edge [-] (f22) (ingress2)
    edge [-] (f23) (f21)
    edge [-] (i4) (f22)
    edge [-] (i4) (f23)
    edge [-] (i4) (ingress3)
    edge [-] (f31) (ingress3)
    edge [-] (f32) (ingress3)
    edge [-] (f33) (f31)
    edge [-] (i5) (f32)
    edge [-] (i5) (f33)
    edge [-] (i5);
  \end{tikzpicture}}
  
  \bigskip

  (b)
  \begin{center}
  \begin{tabular}{ | c | c | c | c | c | }\hline
    Algorithm & 100 Nodes & 250 Nodes & 500 Nodes & 1000 Nodes \\ \hline 
    \textsc{OrderUpdate} & 10 & 83 & 355 & 1101 \\
    \textsc{ConfigPairs} & 129 & 1244 & 3731 & 12077 \\
    \textsc{Refine} & 55 & 267 & Out of memory & Out of memory \\
    \hline
  \end{tabular}
  \end{center}

  \bigskip

  (c)
  \begin{center}
  \begin{tabular}{| c | c | c| c | c| c| c| c| c| c| c| c | c | c |}\hline
    Nodes & 5 & 6 & 7 & 8 & 9 & 10 & 11 & 12 & 13 & 14 & 15 & 30 & 60 \\
    \hline
    Time & 165 & 142 & 166 & 222 & 222 & 205 & 273 & 276 & 354 & 339 & 370 & 611 & 2106 \\
    \hline
  \end{tabular} 
  \end{center}

  \bigskip

  (d)
  \begin{center}
  \begin{tabular}{@{} | c | c | c | c | c | @{}}\hline
    Algorithm & 100 Nodes & 250 Nodes & 500 Nodes & 1000 Nodes \\ \hline
    \textsc{OrderUpdate} & 19 & 170 & 900 & 3963 \\
    \textsc{OrderUpdate} w/o counterexamples & 101 & 1793 & 6269 & Timeout \\
    \textsc{Refine} & 20 & 101 & Out of memory & Out of memory \\
    \hline
  \end{tabular}
  \end{center}
\caption{Experiments: (a) topology, (b) scaling network size, (c)
  scaling update size, (d) impossible updates. All times are in
  seconds.}
\label{fig:experiments}
\end{figure}

\section{Related Work}
\label{sec:related}

Network updates are a form of concurrent programming. Synthesis for
concurrent programs has attracted considerable research attention in
recent years~\cite{SJB08,VYY10,CCG08,VY08}. In work by Solar-Lezama et
al.~\cite{SJB08} and Vechev et al.~\cite{VY08}, an order for a given
set of instructions is synthesized, which is a task similar to
ours. However, the problem settings in the traditional synthesis work
and in this paper are quite different. First, traditional synthesis is
a game against the environment which (in the concurrent programming
case) provides inputs and schedules threads; in contrast our synthesis
problem is a reachability problem on the space of
configurations. Second, the space of network configurations is very
rich; determining whether a configuration is false is an LTL model
checking problem by itself.

Update mechanisms have also been studied in the networking community.
This paper builds on previous work on consistent updates by Reitblatt
et al.~\cite{frenetic-consistent-updates}. However, unlike our tool,
which allows operators to specify explicit invariants, consistent
updates preserve \emph{all} path properties. This imposes a
fundamental overhead as certain efficient updates that are produced by
our tool would not be valid as consistent updates. Another line of
work has investigated update mechanisms that minimize disruptions in
specific routing
protocols~\cite{francois-igp,francois-bgp,raza:11,vanbever:11,kushman:r-bgp}. However,
these methods are tied to particular protocols such as BGP, and only
guarantee basic properties such as connectivity. In particular, they
do not allow the operator to specify explicit invariants.

\section{Conclusion}
\label{sec:conclusion}

Network updates is an area where techniques developed for program and
controller synthesis could be very beneficial for state-of-the-art
systems. There are several possible directions for future work. We
plan to investigate further optimizations that could bring down the
running time on realistic networks from minutes to seconds, improving
usability. We also plan to investigate the network update problem with
environment changing while updates are executed, leading to two-player
games. It would also be interesting to abstract the structure of the
network and apply parametric synthesis techniques, and to explore
techniques that incorporate considerations of network traffic, using
ideas from controller synthesis. Another interesting direction is to
investigate algorithms that rank updates and select the ``best'' one
when there are multiple correct updates. Finally, we would also like
to extend our tool to provide guarantees about properties involving
sets of packets (such as per-flow consistency from Reitblatt et
al.~\cite{frenetic-consistent-updates}), and about properties
concerning bandwidth and other quantitative resources.

\paragraph*{Acknowledgments.} We wish to thank the SYNT 
reviewers, Arjun Guha, and Mark Reitblatt for helpful comments and
suggestions. Our work is supported in part by NSF under grants
CNS-1111698, CCF-1253165, and CCF-0964409; ONR under award
N00014-12-1-0757; by a Google Research Award; and by a gift from Intel
Corporation.

\bibliographystyle{eptcs} 
\bibliography{main}

\end{document}